\documentclass[aps, preprint]{revtex4-1}

\usepackage{graphicx}% 
\usepackage{dcolumn}% 
\usepackage{bm}% 

\usepackage{amsmath}
\usepackage{amssymb}
\usepackage{amsthm}

\newtheorem{prop}{Proposition}[section]

\begin{document}
\title{Bianchi Class B Spacetimes with Electromagnetic Fields}
\date{\today}
\author{Kei Yamamoto}
\email{K.Yamamoto@damtp.cam.ac.uk}
\affiliation{DAMTP, Centre for Mathematical Sciences, University of Cambridge, \\
Wilberforce Road, CB3 0WA, United Kingdom}

\begin{abstract}
We carry out a thorough analysis on a class of cosmological space-times which admit three spacelike Killing vectors of Bianchi class B and contain electromagnetic fields. Using dynamical system analysis, we show that a family of electro-vacuum plane-wave solutions of the Einstein-Maxwell equations is the stable attractor for expanding universes. Phase dynamics are investigated in detail for particular symmetric models. We integrate the system exactly for some special cases to confirm the qualitative features. Some of the obtained solutions have not been presented previously to the best of our knowledge. Finally, based on those analyses, we discuss the relation between those homogeneous models and perturbations of open Friedmann-Lema\^itre-Robertson-Walker universes. We argue that the electro-vacuum plane-wave modes correspond to a certain long-wavelength limit of electromagnetic perturbations.
\end{abstract}


\maketitle

\section{Introduction}
Cosmological magnetic fields have been studied for decades. Various upper bounds have been placed on the strength of any primordial galactic and extragalactic magnetic field by Faraday rotation measure \cite{Kronberg,Wolfe}, CMB anisotropy \cite{Silk,Barrow,Maartens} and primordial nucleosynthesis \cite{Kernan,Giovannini} although they are inconclusive regarding its existence. Recently, a lower bound of $B \sim 10^{-16}$G has been claimed by using TeV gamma-ray sources \cite{Neronov,Ando}. Given the increasing quantity and accuracy of the data from astronomical observations, it is important to examine various theoretical possibilities concerning the large-scale electromagnetism in curved space-time even if the effect is likely to be small.

It is well known that the energy density of electromagnetic field decays as the inverse fourth-power of the scale factor of the universe when it is perturbatively analyzed around a flat Friedmann-Lema\^itre-Robertson-Walker (FLRW) background due to the conformal invariance of the Maxwell equations.  Recently, however, it has been pointed out that in open FLRW universes, the decay of magnetic energy density can be slower than that of blackbody radiation if you take into account modes with wavelengths above a certain threshold \cite{Tsagas}. The source of amplification is the interaction between the scalar curvature of the three-spaces and the vector-mode perturbation of  electromagnetic field, which would be a second-order effect in the perturbation around flat FLRW. Based on that result, it is interesting to look into nonlinear effects of the magneto-curvature coupling on large scales.

Spatially homogeneous Bianchi cosmologies are a suitable framework to generalize FLRW cosmology in the limit of small long-wavelength inhomogeneities. While the model is tractable as a system of ordinary differential equations, it is fully nonlinear and complicated enough to exhibit qualitatively new behaviour. It is known \cite{Hughston} that only types I, II, VI$_0$ and VII$_0$ admit pure magnetic fields among all the different models of spatially homogeneous universes and they have been investigated by many authors \cite{Collins, LeBlancVI, LeBlancI, LeBlancII}. Although the other models have been given less attention in the context of magnetic fields because of the lack of observational evidence for large-scale electric fields, we need to investigate them for the nonlinear effects described above because only type V and VII$_h$ contain open FLRW models as special cases. Indeed, it is this mixing of electric and magnetic mode that gives rise to the amplification in the perturbative analysis. We expect, at least intuitively, that the isotropic limit of those homogeneous universes should reproduce the behaviour of a long-wavelength limit of perturbations around open FLRW. These models belong to a subclass of homogeneous cosmological models, the so-called class B of Bianchi cosmologies which is of interest from a mathematical point of view as well because of the existence of electro-vacuum plane-wave solutions that are generalization of vacuum gravitational plane-waves. In the pure-gravitational class B models, the vacuum plane-wave space-times are found to be stable attractor solutions for expanding initial conditions \cite{ClassB}. When a tilted perfect fluid, whose peculiar velocity is not aligned to the unit normal of the homogeneous spatial slices, is included, it has been shown that they are not necessarily simple attractor solutions and can exhibit non-self-similar looplike behaviour at late time \cite{TiltedB}. Since the inclusion of the electromagnetic field introduces energy fluxes with respect to the homogenous hypersurfaces (as does that of titled fluid), it is natural to ask whether these plane-wave space-times are stable in the Einstein-Maxwell system.

In the present paper, we carry out a thorough analysis of the nonexceptional Bianchi class B cosmological models containing a general electromagnetic field and an orthogonal (nontilted) perfect fluid with a linear equation of state (we refer to it as a $\gamma $-law fluid). First, in section II, we reduce the governing equations to a standard form which is suited to analyze the stability of various self-similar solutions. In section III, the dynamical system analysis suggests that the 2-parameter family of electro-vacuum plane-wave solutions is a stable attractor and possesses features very similar to its pure-gravitational counterpart. In section IV, we take a closer look on the axisymmetric subcases which still contain open FLRW models as a special case. In section V, we integrate the equations analytically in some special cases to complement the qualitative analysis and to facilitate comparison with the perturbative analysis later. Some of the solutions have not been found before to the best of our knowledge. In section VI, we look for a connection between these nonlinear homogeneous models and the large-scale limit of linear perturbations and argue that the Bianchi models correspond to a certain long-wavelength limit of vector perturbations around an open FLRW universe.

\section{Maxwell equations and the dynamical system}\label{sec:ortho}
The notations and terminology are mostly employed from Wainwright and Ellis \cite{Ellis} (chapter 1 for covariant approach to Bianchi cosmologies and chapter 4 for application of dynamical systems analysis to cosmology).  We adopt the convention where four-vectors are presented in bold face and 1-forms come with over-bars. Otherwise all the quantities are understood to be real numbers. Latin indices run from 0 to 3 and Greek letters are used to label the spatial part (1 to 3) of them. Differentiation by proper (clock-) time $t$ is denoted by an overdot. 

 Following Ellis and MacCallum \cite{MacCallum}, we take a group invariant orthonormal frame $\{ \bold{e}_a \}$ and their dual 1-forms $\{ \bar{\omega }^a \} $. The fundamental variables describing the geometry of the space-time are the commutators among the basis vectors
\begin{equation*}
[ \bold{e}_a , \bold{e}_b ] = \gamma ^{c}_{\ ab} \bold{e}_c 
\end{equation*}
and their nonvanishing components are prametrised as
\begin{eqnarray*}
\gamma ^{\alpha }_{\ 0 \beta } &=& -H\delta ^{\alpha }_{\ \beta } - \sigma ^{\alpha }_{\ \beta } -\epsilon ^{\alpha }_{\ \beta \mu }\Omega ^{\mu } , \\
\gamma ^{\alpha }_{\ \beta \gamma } &=& \epsilon _{\beta \gamma \mu }n^{\alpha \mu } +a_{\beta }\delta ^{\alpha }_{\ \gamma } -a_{\gamma }\delta ^{\alpha }_{\ \beta } ,
\end{eqnarray*}
where $\epsilon _{\alpha \beta \gamma }$ is the three-dimensional Levi-Civita symbol and Greek indices are raised/lowered by $\delta ^{\alpha \beta }$/$\delta _{\alpha \beta }$. A proper time coordinate labeling the homogeneous hypersurfaces is defined by
\begin{equation*}
\frac{\partial }{\partial t} = \bold{e}_0 
\end{equation*}
and all the variables above can be seen as functions of $t$ since
\begin{equation*}
\bold{e}_{\alpha } (\gamma ^a_{\ bc}) = 0.
\end{equation*}

As the matter contents of the universe, we take a nontilted $\gamma $-law perfect fluid, whose energy density is denoted by $\rho $, and a source-free electromagnetic field. It is natural, though not always necessary, to demand that all the components of the electromagnetic field in this frame depend only on $t$. In defining the $1+3$ split of field strength tensor $F_{ab}$, we follow the convention of Misner, Thorne and Wheeler \cite{MTW}; that is
\begin{equation*}
\mathfrak{F} = \frac{1}{2}F_{ab}\bar{\omega }^a \wedge \bar{\omega }^b = E_{\alpha }\bar{\omega }^{\alpha } \wedge \bar{\omega }^0 + \frac{1}{2}H^{\alpha }\epsilon _{\alpha \beta \gamma }\bar{\omega }^{\beta } \wedge \bar{\omega }^{\gamma }.
\end{equation*}
The source-free Maxwell equations
\begin{equation*}
d\mathfrak{F} = d^{\ast }\mathfrak{F} =0 
\end{equation*}
can be written in terms of the components in the orthonormal frame with a help of the relation
\begin{equation*}
d \bar{\omega }^a = -\frac{1}{2} \gamma ^a_{\ bc}\bar{\omega }^b \wedge \bar{\omega }^c .
\end{equation*}
The result reads as follows:
\begin{eqnarray}
\dot{H}_{\alpha } &=& -2H H_{\alpha } + \sigma _{\alpha \beta }H^{\beta } + \epsilon _{\alpha \beta \gamma }H^{\beta } \Omega ^{\gamma } + n_{\alpha \beta }E^{\beta } + \epsilon _{\alpha \beta \gamma }a^{\beta } E^{\gamma } , \nonumber \\
\dot{E}_{\alpha } &=& -2HE_{\alpha }  + \sigma _{\alpha \beta } E^{\beta } + \epsilon _{\alpha \beta \gamma } E^{\beta }\Omega ^{\gamma } -n_{\alpha \beta }H^{\beta } -\epsilon _{\alpha \beta \gamma }a^{\beta }H^{\gamma } , \label{eq:div} \\
0 &=& a_{\alpha }E^{\alpha } \ = \ a_{\alpha }H^{\alpha } . \nonumber 
\end{eqnarray}

The isometry group $G_3$ of Bianchi class B admits an Abelian subgroup $G_2$ and we choose $\bold{e}_2$ and $\bold{e}_3$ to be tangent to the orbits of the $G_2$. Then the constraint equations in (\ref{eq:div}) imply $E_1 = H_1 =0$. In the Binachi class B space-time containing this source-free electromagnetic field and the nontilted perfect fluid, it turns out that the $G_2$ acts orthogonally transitively on its orbits unless $h = -\frac{1}{9}$. We focus our attention on this so-called nonexceptional case. Thus we can use the equations given in the section 1.6.3 of Ref. \onlinecite{Ellis}, setting $\bold{\partial }_1 = \dot{u}_ 1 = 0$. The source terms are taken to be
\begin{eqnarray*}
\mu &=& \rho + 3\pi _{+} , \\
p &=& (\gamma -1) \rho + \pi _{+} , \\
\pi _{+} &=& \frac{1}{6}(E_2^2 + E_3^2 + H_2^2 + H_3^2 ) , \\
\tilde{\pi }_{AB} &=& \left(   \begin{array}{cc}
     -\frac{1}{2}(E_2^2 + H_2^2 -E_3^2 -H_3^2 ) & - (E_2 E_3 + H_2 H_3)\\ 
    -(E_2 E_3 + H_2 H_3) & \frac{1}{2}(E_2^2 + H_2^2 -E_3^2 -H_3^2 )  \\ 
  \end{array} \right) , \\
  q_1 &=& E_2 H_3 - E_3 H_2  .
\end{eqnarray*}
We assume $\gamma $ is constant and satisfies $0 <\gamma \leq 2$.

Now we follow a standard procedure established by Wainwright and his collaborators \cite{Wainwright} to reduce the equations into a form suited for qualitative analysis. In doing so, we carry out the $1+1+2$ split of the space-time so that the gauge freedom about the 1-axis is manifest. We classify the expansion normalised variables according to their behaviour under a rotation around the 1-axis. \\

{\it Scalars (spin-0)} :
\begin{equation*}
\begin{array}{ccccccccc}
 \Sigma _{+} & \equiv & \displaystyle \frac{\sigma _{+}}{H}, &   N_{+} & \equiv &  \displaystyle \frac{n_{+}}{H}, & \Pi _{+}  & \equiv &  \displaystyle \frac{\pi _{+}}{H^2 } , \\
\Omega &\equiv &  \displaystyle \frac{\rho }{3H^2 }, &  A  & \equiv &  \displaystyle \frac{a}{H} ,  & \Xi & \equiv &  \displaystyle \frac{q_1 }{3H^2 } .
\end{array}
\end{equation*}

{\it Tensors (spin-2)} :
\begin{equation*}
\begin{array}{ccccccccc}
 \Sigma _{-} &\equiv & \displaystyle \frac{\sigma _{-}}{H},  &  N_{-}  &\equiv &  \displaystyle \frac{n_{-}}{H} ,&  \Pi _{-} & \equiv &  \displaystyle \frac{\tilde{\pi }_{33}}{\sqrt{3}H^2 }, \\
\Sigma _{\times }  & \equiv &  \displaystyle  \frac{\sigma _{\times}}{H}, &  N_{\times } & \equiv &  \displaystyle \frac{n_{\times }}{H}, & \Pi _{\times }  &\equiv  &  \displaystyle -\frac{\tilde{\pi }_{23}}{\sqrt{3}H^2}.
\end{array}
\end{equation*}

{\it Inhomogeneous (rotation angle itself)}:
\begin{equation*}
R \equiv \frac{\Omega _1}{H} .
\end{equation*}
The variables $\Pi _{+,-\times }$ and $\Xi $ are quadratures of electromagnetic field components and their evolution equations can be obtained from the Maxwell's equations (\ref{eq:div}). We introduce a new time coordinate $\tau $ and deceleration parameter $q$ by
\begin{equation*}
\frac{dt}{d\tilde{\tau }} = \frac{1}{H} , \ \ \ \ \ \dot{H} = -(1+q)H^2 ,
\end{equation*}
and use primes to denote derivatives with respect to $\tilde{\tau }$. We derive the following equations:\\

{\it Einstein equations}:
\begin{eqnarray}
q &=& 2(\Sigma _{+}^2 +\Sigma _{-}^2 + \Sigma _{\times }^2 )+\frac{1}{2}(3\gamma -2)\Omega +\Pi _{+}, \label{eq:ein} \\
1 &=& \Omega + \Sigma _{+}^2 + \Sigma _{-}^2 + \Sigma _{\times }^2 + A^2 + N_{-}^2 + N_{\times }^2 + \Pi _{+}, \label{eq:hamilton} \\
0 &=& \Xi + 2(\Sigma _{+}A + \Sigma _{-}N_{\times }-\Sigma _{\times }N_{-}), \label{eq:mom} \\
\Sigma _{+}^{\prime } &=& (q-2)\Sigma _{+}-2(N_{-}^2 +N_{\times }^2 ) -\Pi _{+}, \nonumber \\
\Sigma _{-}^{\prime } &=& (q-2)\Sigma _{-}+2\Sigma _{\times }R-2(N_{+}N_{-}-AN_{\times })-\Pi _{-} ,\\
\Sigma _{\times }^{\prime } &=& (q-2)\Sigma _{\times }-2\Sigma _{-}R -2(N_{+}N_{\times }+AN_{-})-\Pi _{\times } .\nonumber
\end{eqnarray}

{\it Jacobi identities}:
\begin{equation}
\begin{array}{ccl}
N_{+}^{\prime } &=& (q+2\Sigma _{+})N_{+}+ 6(\Sigma _{-}N_{-}+\Sigma _{\times }N_{\times }), \\
N_{-}^{\prime } &=& (q+2\Sigma _{+})N_{-}+2(\Sigma _{-}N_{+}+N_{\times }R) ,\\
N_{\times }^{\prime } &=& (q+2\Sigma _{+})N_{\times }+2(\Sigma _{\times }N_{+}-N_{-}R) ,\\
A^{\prime } &=& (q+2\Sigma _{+})A  .
\end{array} 
\end{equation}

{\it Class B first integral}:
\begin{equation}
A^2 = h\{ N_{+}^2 -3(N_{-}^2 +N_{\times }^2) \} . \label{eq:jacob}
\end{equation}

{\it Energy conservation for the fluid}:
\begin{equation}
\Omega ^{\prime } = (2q-3\gamma +2)\Omega .
\end{equation}

{\it Maxwell's equations}:
\begin{equation}
\begin{array}{ccl}
\Pi _{+}^{\prime } &=& 2(q-1+\Sigma _{+})\Pi _{+}+2(\Sigma _{-}\Pi _{-}+\Sigma _{\times }\Pi _{\times } +A \Xi ) , \\
\Xi ^{\prime } &=& 2(q-1+\Sigma _{+})\Xi +2(A\Pi _{+}+N_{\times }\Pi _{-}-N_{-}\Pi _{\times }) , \\
\Pi _{-}^{\prime } &=& 2(q-1+\Sigma _{+})\Pi _{-} + 6\Sigma _{-}\Pi _{+}-6N_{\times }\Xi +2R\Pi  _{\times } , \\
\Pi _{\times }^{\prime } &=& 2(q-1+\Sigma _{+})\Pi _{\times } +6\Sigma _{\times }\Pi _{+} +6N_{-}\Xi -2R\Pi _{-} . 
\end{array} \label{eq:quad}
\end{equation}
There is an algebraic constraint among the quadratures:
\begin{equation}
\Pi _{+}^2 =  \Xi ^2 + \frac{1}{3}(\Pi _{-}^2 + \Pi _{\times }^2 )  \label{eq:pi}
\end{equation}
 
Here we comment on the degrees of freedom. The dynamical system is not closed because there is a remaining gauge freedom as we mentioned above. In order to decide the frame to fix the gauge, we face a similar problem as was seen in the analysis of tilted Bianchi class B models by Coley and Hervik \cite{TiltedB}. Namely, that there is not a simple choice of frame for which the governing equations become regular everywhere. Which choice is convenient depends on the equilibrium points or invariant sets under consideration and we will make use of several different choices in the present article. Note that construction of scalar variables such as $\Sigma _{-}N_{-} + \Sigma _{\times }N_{\times }$ would lead to too many constraints as was the case in the Ref. \onlinecite{TiltedB}. Whichever choice we make, however, the dimension of the dynamical system is seven. The counting goes as follows; we take $\Sigma _{+,-,\times}, N_{+,-,\times }$ and $\Pi _{+,-,\times }$ as the fundamental variables; $A$, $\Xi $ and $\Omega $ are determined by (\ref{eq:jacob}), (\ref{eq:pi}) and the Hamiltonian constraint (\ref{eq:hamilton}) respectively; $R$ will be determined by fixing the gauge, which also imposes an additional constraint among the spin-2 variables; finally taking into account the momentum constraint (\ref{eq:mom}), we are left with nine variables and two constraints.

\section{Stability of plane-wave solutions}
The purpose of this section is to extend the result obtained by Hewitt and Wainwright \cite{ClassB} about the stability of vacuum gravitational waves to Einstein-Maxwell equations. We follow the usual procedure of dynamical system analysis.

\subsection{Invariant sets}
To facilitate the understanding of the dynamical system, it is worth sorting out possible {\it invariant sets}, which are defined to be subsets consisting of trajectories specified by certain restrictions which form dynamical systems by themselves. Note that we may assume $A \geq 0$ without loss of generality because of the reflection symmetry about the 2-3 plane. We use the integration constant $\tilde{h} \equiv 1/h$ of (\ref{eq:jacob}) to classify the class B models.
\begin{table}[htdp]
\caption{Electromagnetic Bianchi invariant sets. The entries in the upper half are the general class B models for the source-free electromagnetic field.}
\begin{ruledtabular}
\begin{tabular}{llc}
Notation & Restrictions & Dimension \\ \hline
$M({\rm VI}_h )$ & $\tilde{h}<0, A>0$ & $7$ \\ 
$M({\rm VII}_h )$ & $\tilde{h}>0 , A>0 , N_{+}>0$ & $7$ \\
$ M({\rm IV}) $ & $\tilde{h}=0 , A>0, N_{+} >0,  N_{+}^2 = 3(N_{-}^2 +N_{\times }^2 ) $ & $7$ \\
$M({\rm V})$ & $\tilde{h}=0, A>0, N_{+}=N_{-}=N_{\times }=0$ & $5$ \\ \hline 
$M({\rm II})$ & $A=0,  N_{+}^2 =3(N_{-}^2 +N_{\times }^2 ) $ & $6$ \\
$M({\rm I})$& $A=N_{+}=N_{-}=N_{\times }= \Xi = 0 $ & $4$ \\
\end{tabular}
\end{ruledtabular}
\label{ClassB}
\end{table}%

Table \ref{ClassB} lists the electromagnetic Bianchi invariant sets described by the equations given in the previous section. The restriction on the sign of $N_{+}$ in type VII$_h$ and IV comes from the fact that $N_{+}=0$ by itself forces the system to fall into an invariant set and there is a discrete symmetry interchanging 2- and 3-directions. We refrain from further discussion for disconnected components as it is subtle before the gauge is fixed. $M({\rm V})$ forms a part of the boundary of $M({\rm IV})$.  $M({\rm II})$ is the class A boundary for all VI$_h$, VII$_h$ and IV while it is not the most general electromagnetic type II. Its four dimensional subset and $M({\rm I})$ were investigated by LeBlanc \cite{LeBlancII, LeBlancI}. 

\begin{table}[htdp]
\caption{Electromagnetic Bianchi class B invariant sets with higher symmetry.}
\begin{ruledtabular}
\begin{tabular}{llc}
Notation & Restrictions & Dimension \\ \hline
$SM({\rm VI}_h )$ & $\tilde{h}<0, A>0, N_{+}=0$ & $4$ \\ 
$SM({\rm VII}_h )$ & $\tilde{h}>0 , \Sigma _{-,\times } = N_{-,\times } = \Pi _{-,\times } = 0, A>0, N_{+}>0 $ & $2$ \\
$ SM({\rm V}) $ & $\tilde{h} =0 , \Sigma _{-,\times } = N_{\pm ,\times } = \Pi _{-,\times } = 0 , A>0 $ & $2$ \\
\end{tabular}
\end{ruledtabular}
\label{symmetric}
\end{table}%

Some symmetric invariant sets are shown in Table \ref{symmetric}. $SM({\rm VI}_h)$ has four dimensions and can support non-null Maxwell fields. $SM({\rm VII}_h)$ and $SM({\rm V})$ have an identical two-dimensional structure as dynamical systems and the space-time is locally rotationally symmetric (LRS) \cite{LRS}. They will play a central role in the comparison to perturbations around open FLRW models because we can fully derive the global behaviour of the orbits and integrate them explicitly for some interesting cases. 

Aside from the invariant sets listed above, there are obvious electro-vacuum subsets of them obtained by $\Omega =0$. Here we are not concerned with pure-gravitational orthogonal class B models (i.e. $\Pi _{+}=0$) as they were studied by Hewitt and Wainwright \cite{ClassB}.

\subsection{Equilibrium points}

In dynamical systems analysis, {\it equilibrium points}, which are defined as time-independent solutions of the system, play an important role. Here we focus on two of them which are important concerning the past- and future-asymptotic behaviour of the system. The full list of the equilibrium points in the present setup is given in table \ref{equilibrium}.
\begin{table}[htdp]
\caption{Equilibrium points in the electromagnetic Bianchi class B. Some of the symbols represent several distinct points in the state space because they lie in disconnected boundaries with a same structure. For the notation of the invariant sets, see Hewitt and Wainwright \cite{ClassB}}
\begin{ruledtabular}
\begin{tabular}{lclc}
Symbol & Invariant set & Self-similar solution & Restrictions \\ \hline
$P({\rm I})$ & $B({\rm I})$ & Flat FLRW &  \\
$P({\rm II})$ & $B({\rm II})$ & Collins-Stewart & $\frac{2}{3} <\gamma < 2$ \\
$P({\rm VI }_h)$ & $B({\rm VI} _h )$& Collins & $ h>\frac{\gamma -2}{3\gamma -2}$, $\frac{2}{3} <\gamma < 2$ \\
$PM_1 ({\rm I})$ & $M({\rm I})$ & Jacobs magnetic I & $\frac{4}{3} <\gamma < 2$ \\
$PM_2 ({\rm I})$ & $M({\rm I})$ & LeBlanc & $\frac{8}{5} <\gamma < 2$ \\
$PM({\rm II})$ & $M({\rm II})$ & Dunn-Tupper & $\frac{6}{5} <\gamma < 2$ \\
$\mathcal{K}$ & $B({\rm I})$ & Kasner vacuum & \\
$\mathcal{J}$ & $B({\rm I})$ & Jacobs stiff fluid & $\gamma =2 $ \\
$\mathcal{P}M(\tilde{h})$ &  $ \begin{array}{l}
   M({\rm VI}_h) \\ 
    M({\rm VII}_h) \\ 
    M({\rm IV}) \\ 
  \end{array}$ & Electro-vacuum plane-waves & $  \begin{array}{c}
     \tilde{h}<0 \\ 
    \tilde{h}>0 \\ 
    \tilde{h}=0 \\ 
  \end{array} $ \\
$M(\tilde{h})$ & $  \begin{array}{c}
    SM({\rm VII}_h) \\ 
    SM({\rm V}) \\ 
  \end{array}$ & Milne universe & $  \begin{array}{c}
    \tilde{h}>0 \\ 
    \tilde{h}=0 \\ 
  \end{array}$ \\
$\mathcal{F}(\tilde{h}) $ &$  \begin{array}{c}
    SM({\rm VII}_h) \\ 
    SM({\rm V})\\ 
  \end{array}$ & Open FLRW & $  \begin{array}{c}
     \tilde{h}>0 \\ 
     \tilde{h}=0 \\ 
  \end{array}$, $\gamma = \frac{2}{3}$  \\ 
  \end{tabular}
\end{ruledtabular}
\label{equilibrium}
\end{table}%

\begin{description}
\item[ Kasner equilibrium points  $\mathcal{K}$ ]

Those are a one-parameter family of equilibrium points corresponding to the well-known Kasner vacuum solutions. They have always been the key to understand past-asymptotic behaviour in Bianchi cosmologies and we expect the same is true here. To carry out the stability analysis, it is convenient to fix the gauge by setting $\Sigma _{\times }=0$. $R$ is determined by
\begin{equation*}
\Sigma _{-}R = -N_{+}N_{\times } - AN_{-} - \frac{1}{2}\Pi _{\times } .
\end{equation*}
In the resulting system of equations, $\mathcal{K}$ is characterised as follows:
\begin{eqnarray*}
\Sigma _{+} = \cos \psi \ \ \ \ \ \Sigma _{-} = \sin \psi \ \ \ \ \  -\pi < \psi \leq \pi  \\
A= N_{\pm ,\times } = \Pi _{+}  = \Omega =  0.
\end{eqnarray*}
This can be pictured as a unit circle on $\Sigma _{+}$-$\Sigma _{-}$ plane with $\psi $ measuring the angle. Since all the internal curvatures ($A, N_{\pm, \times }$) and electromagnetic quadratures are zero, the constraints (\ref{eq:jacob}) and (\ref{eq:pi}) become degenerate. Therefore we have nine eigenvalues as follows:
\begin{eqnarray*}
\lambda _1 = 0 \ \ \ \ \ \lambda _2 = -3(\gamma -2) \ \ \ \ \ \lambda _{3,4} = 2(1+\cos \psi \pm \sqrt{3}\sin \psi ) \\
\lambda _{5,6,7} = 2(1+\cos \psi ) \ \ \ \ \ \lambda _{8,9} = 1+\cos \psi \pm \sqrt{3}\sin \psi .
\end{eqnarray*}
It is difficult to decide which of them represent the true dynamical degrees of freedom. However, as we can see, the arc specified by $-\frac{\pi }{3} < \psi < \frac{\pi }{3}$, which we shall call $\tilde{\mathcal{K}}$, has positive eigenvalues for all the directions aside from the zero around the circumference. The rest of $\mathcal{K}$ except for a finite number of points have a stable manifold of at least one dimension represented by either Rosen orbits, Taub orbits, or rotating Kasner orbits studied by LeBlanc for magnetic type I \cite{LeBlancI} and type II \cite{LeBlancII}. This past-stability of the arc $\tilde{\mathcal{K}}$ also suggests that the Mixmaster oscillation will not occur in this class of models. We denote the points $\psi = 0$ by $Q_1$ and $\psi = \pi $ by $T_1$. They will appear in the LRS invariant sets considered later.

\item[Electromagnetic plane-wave equilibrium points $\mathcal{P}M_h $]
They are essentially the only nontrivial electromagnetic class B equilibrium points and of interest since the pure-gravitational counterpart (a subset of the electromagnetic ones) is a stable attractor of the expanding class B space-times \cite{ClassB}. The electromagnetic field is null and the corresponding self-similar solutions were discussed by Harvey et.al. \cite{Harvey}, Araujo and Skea \cite{Araujo}, and Hervik \cite{Hervik} in different contexts. Here we use the gauge freedom to take $N_{\times }=0$, which implies
\begin{equation*}
R= \frac{N_{+}}{N_{-}}\Sigma _{\times } .
\end{equation*}
They are expressed as a plane of equilibria, namely a two-parameter family of points characterised by constants $r$ and $s$. 
\begin{eqnarray*}
&& \Sigma _{+} = -r \ \ \ \ \ \Sigma _{\times } = -N_{-} = s \ \ \ \ \ \Sigma _{-} =0 \ \ \ \ \ N_{+}^2 = \tilde{h} (1-r )^2 +3s^2 \\
&& A = 1-r \ \ \ \ \  \Pi _{+} = \Xi  = 2r(1-r) -2s^2 \ \ \ \ \ \ \Pi _{-} =\Pi _{\times } = \Omega =0  \\
&& 0 \leq  r\leq 1 \ \ \ \ \ s^2 \leq r(1-r)  .
\end{eqnarray*}
For VI$_h$, we have an additional restriction
\begin{equation*}
-\frac{\tilde{h}}{3}(1-r)^2 \leq s^2 .
\end{equation*}
Thus the area of the plane becomes smaller for greater $h$. When $s^2 = r(1-r)$, they correspond to the pure-gravitational plane-wave space-times. The eigenvalues are given by
 \begin{eqnarray*}
&& \lambda _1 = 4r-3\gamma +2 \ \ \ \ \ \lambda _2 = \lambda _3 = 0  \\
&& \lambda _{4,5,6,7} = 2(r-1) \pm \sqrt{-2(5N_{+}^2 + 3\Pi _{+})\pm 6N_{+}\sqrt{N_{+}^2 + 6\Pi _{+}}}
\end{eqnarray*}
where the square root always gives pure imaginary values. The two zero eigenvalues reflect the fact that the equilibrium points form a two-dimensional subset. For type VII$_h$ and IV with $\frac{2}{3}<\gamma < 2$, they always have a stable subset specified by $r<\frac{1}{4}(3\gamma -2)$. For type VI$_h$, this condition cannot be satisfied if $ h>\frac{\gamma -2}{3\gamma -2}$, in which case $P({\rm VI}_h)$ becomes the attractor of the dynamical system. These features are very much analogous to those found for the pure-gravitational case. Note that the direction of $\lambda _1$ corresponds to the perturbation in $\Omega $. Therefore, $\mathcal{P}M(\tilde{h})$ is a local sink in the electro-vacuum ($\Omega =0$) class B models. 
 \end{description}

\subsection{Monotone function}

The condition $\Omega =0$ gives rise to six-dimensional boundaries of the full electromagnetic class B models. Analogous to the vacuum Einstein models \cite{ClassB}, we can show that $Z \equiv (1+\Sigma _{+})^2 -A^2 $ is monotonic:
\begin{equation*}
Z^{\prime } = 2(q-2)Z -3(\gamma -2)(1+\Sigma _{+})\Omega 
\end{equation*}
which implies $|Z|$ is monotone decreasing for $\Omega =0$ since $q< 2$ in the interior of $M({\rm IV})$, $M({\rm VI}_h)$ and $M({\rm VII}_h)$. In fact, $Z$ itself is monotone decreasing since $Z \geq 0$. To see this, first note that (\ref{eq:hamilton}) $\pm $ (\ref{eq:mom}) yield
\begin{equation*}
1-\Omega = (\Sigma _{+}\pm A)^2 + (\Sigma _{-} \pm N_{\times })^2 + (\Sigma _{\times }\mp N_{-})^2 +\Pi _{+} \pm \Xi .
\end{equation*}
Taking into account (\ref{eq:pi}), $|\Xi | \leq \Pi_{+} $ and therefore $|\Sigma _{+} \pm A | \leq 1$. Then it follows that
\begin{equation*}
Z = (1+\Sigma _{+}+A)(1+\Sigma _{+}-A) \geq 0 .
\end{equation*}
Thus, we obtain the following result, which is a direct generalisation of the proposition 5.1 in Ref. \onlinecite{ClassB}.
\begin{prop}
\item For the vacuum Einstein-Maxwell Bianchi class B models, i.e. $\displaystyle M({\rm IV})\big| _{\Omega = 0}$, $\displaystyle M({\rm VI}_h)\big| _{\Omega = 0}$ and $\displaystyle M({\rm VII}_h)\big| _{\Omega = 0}$ except for a set of measure zero, the past attractor is $\displaystyle \tilde{\mathcal{K}}$ and the future attractor is $\mathcal{P}M(\tilde{h})$. 
\end{prop}
\begin{proof}
  To obtain an everywhere well-defined bounded dynamical system, we partially fix the gauge by setting $R=0$. Then we can apply LaSalle's invariance principle and see that the past- and future-asymptotic sets are contained in the subset $Z^{\prime }=0$. Since Z is monotone decreasing, the past and future are specified by $q=2$ and $Z=0$ respectively. It is easy to see that $q=2$ uniquely characterises the Kasner circle $\mathcal{K}$ and the condition $(1+\Sigma _{+})^2 =A^2 $ is consistent with the plane-wave solutions $\mathcal{P}M(\tilde{h})$ only. Note that these characterisations are gauge-invariant. Therefore, we conclude that the future attractor is $\mathcal{P}M(\tilde{h})$ as it is always stable in the vacuum invariant sets. The past attractor is contained in $\mathcal{K}$. We can specify it to be $\tilde{\mathcal{K}}$ since for all the points in $\mathcal{K} \backslash \overline{\tilde{\mathcal{K}}}$ except $T_1 $, the only zero-eigenvalue is the one in the direction along the circumference. Then the orbits emanating from $\mathcal{K} \backslash \overline{\tilde{\mathcal{K}}}$ lie in its unstable manifold, therefore are at most of measure zero (e.g. see Theorem 4.1 in Aulbach \cite{Aulbach}). $T_1$ cannot be a past attractor as $T_1 \in \overline{\mathcal{P}M(\tilde{h})}$.
  \end{proof}
  
  \section{LRS models}
 The LRS assumption leads to a two-dimensional dynamical system $SM({\rm V})$ or $SM({\rm VII}_h)$, both of which contain the open FLRW model. Only the null Maxwell field is consistent with the geometry. We visualise the dynamics to facilitate the comparison to perturbations around FLRW later. The equations are given by
 \begin{eqnarray*}
 \Sigma _{+}^{\prime } &=& (q-2)\Sigma _{+} -\Pi _{+}, \\
 A^{\prime } &=& (q+2\Sigma _{+})A , \\
 q &=& \frac{1}{2}(3\gamma -2)(1-A^2 ) -\frac{1}{2}(3\gamma -6)\Sigma _{+}^2 -\frac{1}{2}(3\gamma -4)\Pi _{+}, \\
 \Pi _{+}^2 &=& 4A^2 \Sigma _{+}^2 ,  \label{eq:momentum} \\
 \Omega &=& 1 - A^2 -\Sigma _{+}^2 -\Pi _{+}, \\
 N_{+} &=& \sqrt{\tilde{h}}A .
 \end{eqnarray*}
The fourth equation was obtained by combining the momentum constraint with (\ref{eq:pi}). The physical region in the $A$-$\Sigma _{+}$ plane is defined by
\begin{equation*}
0<A<1,  \ \ \ \ \ \Pi _{+} >0, \ \ \ \ \ \Omega >0 
\end{equation*}
and its boundaries. There are two disconnected invariant sets $\Sigma _{+} >0$ and $\Sigma _{+}<0$ separated by the open FLRW orbit $\Sigma _{+} = \Pi _{+} =0$. Within each of their interior, we can solve the quadratic constraint as $\Pi _{+} = \pm 2A\Sigma _{+}$. Substituting this into the Hamiltonian constraint and using $\Omega >0$ yield
\begin{equation*}
A+\Sigma _{+}<1 , \ \ \ \ \ A-\Sigma _{+}<1 .
\end{equation*}
The projection of $\mathcal{P}M(\tilde{h})$ onto the LRS invariant sets is the line $A-\Sigma _{+} = 1$. At the ends of the line are located $M(\tilde{h})$ and $T_1$. They have one zero-eigenvalue along the line and another $4r-3\gamma +2$ representing the perfect fluid mode. Thus the points with $r > \frac{1}{4}(3\gamma -2) $ are local sources and those with $r <\frac{1}{4}(3\gamma -2)$ are local sinks; increasing $\gamma $ means increasing the unstable part in the line $A-\Sigma _{+} =1$. The other source and sinks are summarised in table \ref{fig:asympt}.

  \begin{table}[h]
\caption{\label{fig:asympt}The list of local sources and sinks in the LRS models. The system shows different behaviours according to the direction of Poynting vector.}
\begin{tabular}{c|c|c|c|c|c|c|c|c} \hline \hline
Invariant set &  \multicolumn{4}{c|}{$\Sigma _{+}>0$} & \multicolumn{4}{c}{$\Sigma _{+}<0$ } \\ \hline
 $\gamma $ & $ [ 0, \frac{2}{3} ) $ &$\frac{2}{3} $ & $(\frac{2}{3} , 2) $ & $2$ & $ [ 0, \frac{2}{3} ) $ &$\frac{2}{3} $ & $(\frac{2}{3} , 2) $ & $2$ \\ \hline
 Local sources &  \multicolumn{3}{|c|}{$Q_1$} & $\mathcal{J}$ & \multicolumn{3}{|c|}{$\mathcal{P}M(\tilde{h})$} & $\mathcal{J}$ \\ \hline
 Local sinks & $P({\rm I})$ & $\mathcal{F}(\tilde{h}) $ & \multicolumn{2}{|c|}{$M(\tilde{h}) $} & $P({\rm I})$ & $\mathcal{F}(\tilde{h})$ &\multicolumn{2}{|c}{$\mathcal{P}M(\tilde{h})$ } \\ \hline \hline
 \end{tabular}
 \end{table}

We can find monotonic functions for both of those invariant sets $\Sigma _{+} \gtrless 0$ when $\gamma \geq \frac{2}{3}$.
\begin{equation*}
A^{\prime } = (q+2\Sigma _{+})A >0 \ \ \ \ {\rm for}\ \ \ \ \Sigma _{+}>0,
\end{equation*}
\begin{equation*}
(\alpha A -\beta \Sigma _{+})^{\prime } = \left[ \alpha ^2 A -\beta ^2 \Sigma _{+} + (\alpha  A +\beta \Sigma _{+})^2 \right] \left[ 1-(A-\Sigma _{+})\right] >0 \ \ \ \ {\rm for } \ \ \ \ \Sigma _{+}<0,
\end{equation*}
\begin{equation*}
\alpha \equiv \frac{1}{2}(3\gamma -2) , \ \ \ \ \ \ \beta \equiv \frac{1}{2}(3\gamma -6) .
\end{equation*}
Thus, by the LaSalle's invariance principle, the local sources are the past attractors and the local sinks are the future attractors. The phase portrait of a representative case ($\gamma =1$) is given in figure \ref{fig:dust}.

  The case $\Sigma _{+}<0$ with $\frac{2}{3} <\gamma < 2$ is worth elucidating. Since two distinct orbits never meet each other, we infer that an orbit started from a more anisotropic state (larger $|\Sigma _{+}|$) than the other ends up in a less anisotropic future-asymptotic state. For those models, both past and future attractors have non-zero $\Pi _{+}$ but they isotropise and become close to flat at intermediate times. 
  
   \begin{figure} [h]
 \includegraphics[scale=0.20]{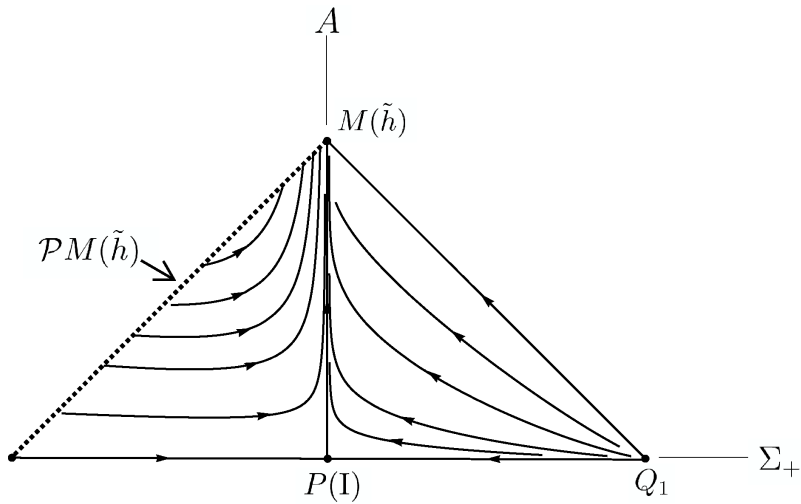}
 \caption{\label{fig:dust}Phase portrait of type V and VII$_h$ LRS models for dust ($\gamma =1$). The diagonal dotted line represents the projection of equilibrium points $\mathcal{P}M_h$ onto the LRS invariant set.}
 \end{figure}

\section{Exact Solutions for LRS space-times}

We are going to present some exact solutions, including new ones, for the dynamical system considered in the previous section, which should reproduce the asymptotic behavior described there. 

In the parametrisation introduced by Siklos \cite{Siklos}, the line element for LRS type V and VII$_h$ is given by
\begin{equation*}
ds^2 = -a^{-2}d\tau ^2 + a^{-2} dx^2  + e^{2(\lambda -x)}(dy^2 + dz^2)\ \ \ \ \ \ \ a>0 .
\end{equation*}
Defining new variables
\begin{equation}
\alpha \equiv \frac{1}{a}\frac{da}{d\tau }, \ \ \ \ \ \ \ \beta \equiv \frac{d\lambda }{d\tau }, \label{eq:tau}
\end{equation}
the Einstein equations reduce to
\begin{eqnarray}
\frac{d\alpha }{d\tau }-2\frac{d\beta }{d\tau }-2\alpha \beta -2 \beta ^2 &=& \frac{1}{2}(3\gamma -2)a^{-2}\rho + 3a^{-2}\pi _{+} , \label{eq:1} \\
2(\alpha +\beta ) &=& a^{-2 }q_1 , \label{eq:2} \\
-\frac{d\alpha }{d\tau } -2\alpha \beta -2 &=& \frac{1}{2}(2-\gamma )a^{-2}\rho + 3a^{-2}\pi _{+}  ,\label{eq:3} \\
\frac{d\beta }{d\tau }+2\beta ^2 -2 &=& \frac{1}{2}(2-\gamma )a^{-2 }\rho \label{eq:4} .
\end{eqnarray}
Combining eqns. $ (\ref{eq:1})+(\ref{eq:3}) +2\times (\ref{eq:4}) $ gives the generalized Friedmann equation
\begin{equation}
\beta ^2 - 2\alpha \beta -3 = a^{-2}(\rho +3\pi _{+} ) \label{eq:5} .
\end{equation}
The energy density of the electromagnetic field $\rho _{em}$ is given by $\rho_{em} = 3\pi _{+}$. Note that $\rho _{em}^2 = q_1 ^2 $ because only null-field is consistent with the LRS symmetry. We term $q_1 = \rho _{em}$ be left-handed and $q_1 = -\rho _{em}$ be right-handed.

\subsection{Left-handed solutions}
\begin{description}
\item[ Eelctro-vacuum ] 
\begin{eqnarray*}
\lambda ^2 ds^2 = e^{2\nu \tau }(-d\tau ^2 + dx^2 ) +e^{2(\tau -x)}(dy^2 + dz^2 )
\end{eqnarray*}
\begin{equation*}
-\infty < \tau < \infty \ \ \ \ \ \ \ \nu >1 
\end{equation*}
\begin{equation*}
 {\rm Alternative \ metric } \ \ \ : \ \ \ ds^2 = -dt^2 + (\nu t)^2 dx^2 +(\nu  t)^{\frac{2}{\nu }}e^{-2x}(dy^2 +dz^2 )
\end{equation*}
This is just LRS specialisation of $\mathcal{P}M(\tilde{h})$ with $s =0$. The parameter $\nu $ is related to $r$ in the section 3 by
\begin{equation}
r = \frac{\nu -1}{\nu +2} . \label{eq:nu}
\end{equation}

\item[ Perfect fluid $\gamma = \frac{2}{3} $  ] 
\begin{eqnarray*}
\lambda ^2 ds^2 = \frac{e^{2r\tau }}{\sqrt{1+e^{2(1-r)\tau }}}(-d\tau ^2 +dx^2 )+ e^{2(r\tau -x)}(1+e^{2(1-r)\tau })(dy^2 +dz^2)
\end{eqnarray*}
\begin{equation*}
-\infty < \tau < \infty \ \ \ \ \ \ \ r>1 
\end{equation*}

{\it Asymptotic behaviour } : 
\begin{eqnarray*}
\tau \rightarrow -\infty & : & ds^2 = -dt^2 + \left( \frac{3r-1}{2}t\right) ^2 dx^2 + \left( \frac{3r-1}{2}t \right) ^{\frac{4}{3r-1}}e^{-2x}(dy^2 +dz^2 ) 
\end{eqnarray*}
\begin{eqnarray*}
 \tau \rightarrow +\infty & : & ds^2 = -dt^2 + (rt)^2 \left[ dx^2 + e^{-2x}(dy^2 + dz^2 ) \right]
\end{eqnarray*}

\item[ Perfect fluid $\gamma =2 $  ] 
\begin{eqnarray*}
\lambda ^2 ds^2 = e^{4r\tau } \sinh ^{1+2r}2\tau (-d\tau ^2 +dx^2 )+\sinh 2\tau e^{-2x}(dy^2 +dz^2)
\end{eqnarray*}
\begin{equation*}
0<\tau < \infty \ \ \ \ \ \ \ r >0
\end{equation*}

{\it Asymptotic behaviour } :
\begin{eqnarray*}
\tau \rightarrow 0 & : & ds^2 = -dt^2 + t^{\frac{2(1+2r)}{3+2r}}dx^2 + t^{\frac{2}{3+2r}}e^{-2cx}(dy^2 + dz^2 ) 
\end{eqnarray*}
\begin{eqnarray*}
 \tau \rightarrow \infty & : & ds^2 = -dt^2 + \left[ (4r+1)t\right] ^2 dx^2 + \left[ (4r+1)t\right] ^{\frac{2}{1+4r}}e^{-2x}(dy^2 + dz^2 )
\end{eqnarray*}

\item[Perfect fluid $\gamma = \frac{4}{3}$ ]
\begin{eqnarray*}
 \lambda ^2 ds^2 = \frac{e^{4r\eta }}{r-\tanh \eta }\left[ -\frac{(r^2-1)^2 d\eta ^2}{4(r\cosh \eta -\sinh \eta )^2} + \frac{dx^2}{\cosh ^2 \eta } \right] + \frac{e^{r\eta -2x}\cosh \eta }{r-\tanh \eta }(dy^2 +dz^2 )
\end{eqnarray*}
\begin{equation*}
-\infty < \eta < \infty \ \ \ \ \ \ \ r>1 
\end{equation*}
The time coordinate $\eta $ is related to $\tau $ by
\begin{equation*}
\frac{d\eta }{d\tau } = r-\sqrt{r^2 -1}\tanh \left( 2\sqrt{r^2 -1}\eta \right) .
\end{equation*}
{\it Asymptotic behaviour } : 
\begin{eqnarray*}
 \eta \rightarrow - \infty & : & ds^2 = -dt^2 + \left( \frac{2(2r+1)}{r-1}t \right) ^2 dx^2 + \left( \frac{2(2r+1)}{r-1}t\right) ^{\frac{r-1}{2r+1}}e^{-2x}(dy^2 + dz^2 ) 
\end{eqnarray*}
\begin{eqnarray*}
 \eta \rightarrow +\infty & : & ds^2 = -dt^2 + \left( \frac{2(2r-1)}{r+1}t \right) ^2 dx^2 + \left( \frac{2(2r-1)}{r+1}t \right) ^{\frac{r+1}{2r-1}}e^{-2x}(dy^2 + dz^2 )
\end{eqnarray*}

\end{description}

\subsection{Right-handed solutions}
\begin{description}
\item[ Electro-vacuum ] 
\begin{eqnarray*}
\lambda ^2 ds^2 = e^{3\tau }\sinh ^{-\frac{1}{2}}2\tau ( -d\tau ^2 + dx^2 ) + \sinh 2\tau e^{-2x}(dy^2 + dz^2 )
\end{eqnarray*}
\begin{equation*}
0 < \tau < \infty
\end{equation*}

{\it Asymptotic behaviour } : 
\begin{eqnarray*}
\tau \rightarrow 0 & : & ds^2 = -dt^2 + t^{-\frac{2}{3}}dx^2 +t^{\frac{4}{3}}e^{-2cx}(dy^2 +dz^2 ) 
\end{eqnarray*}
\begin{eqnarray*}
\tau \rightarrow \infty & :& dx^2 = -dt^2 + t^2 \left[ dx^2 +e^{-2x}(dy^2 + dz^2 ) \right] 
\end{eqnarray*}

\item[ Perfect fluid $\gamma = \frac{2}{3}$  ] 
\begin{eqnarray*}
 \lambda ^2 ds^2 = \frac{e^{2r\tau }}{\sqrt{1-e^{-2(r+1)\tau }}}(-d\tau ^2 +dx^2 ) +e^{2(r\tau -x)}(1-e^{-2(r+1)\tau })(dy^2 +dz^2 )
\end{eqnarray*}
\begin{equation*}
0 < \tau < \infty \ \ \ \ \ \ \ r>1 
\end{equation*}

{\it Asymptotic behaviour } : 
\begin{eqnarray*}
\tau \rightarrow 0 & : & ds^2 = -dt^2 + t^{-\frac{2}{3}}dx^2 + t^{\frac{4}{3}}e^{-2cx}(dy^2 + dz^2 ) 
\end{eqnarray*}
\begin{eqnarray*}
\tau \rightarrow \infty & : & ds^2 = -dt^2 + (rt )^2 \left[ dx^2 + e^{-2x}(dy^2 + dz^2 ) \right]
\end{eqnarray*}

\item[ Perfect fluid $\gamma = 2$  ] 
\begin{eqnarray*}
\lambda ^2 ds^2 = e^{4r\tau }\sinh ^{1-2r} 2\tau (-d\tau ^2 + dx^2 )+\sinh 2\tau e^{-2x}(dy^2 +dz^2 )
\end{eqnarray*}
\begin{equation*}
0 < \tau < \infty \ \ \ \ \ \ \  0<r<\frac{3}{4}
\end{equation*}

{\it Asymptotic behaviour } :
\begin{eqnarray*}
\tau \rightarrow 0 & : & ds^2 = -dt^2 + t^{\frac{2(1-2r)}{3-2r}}dx^2 + t^{\frac{2}{3-2r}}e^{-2cx}(dy^2 + dz^2 ) 
\end{eqnarray*}
\begin{eqnarray*}
\tau \rightarrow \infty & : & ds^2 = -dt^2 + t^2 \left[ dx^2 + e^{-2x}(dy^2 + dz^2 ) \right]
\end{eqnarray*}

\item[ Perfect fluid $\gamma = \frac{4}{3}$ ]
\begin{eqnarray*}
 \lambda ^2 ds^2 =  e^{8r\eta } A(\eta )^{-2}B(\eta ) \left( -B(\eta )^{2}d\eta ^2 + dx^2 \right) + e^{2r\eta -2x}A(\eta )B(\eta ) \left( dy^2 +dz^2 \right)
\end{eqnarray*}
\begin{equation*}
A(\eta ) =     \frac{\sin \left( 2\sqrt{1- r^2 } \eta \right) }{\sqrt{1- r^2 }} \ \ \ B(\eta ) = 
    \frac{1}{\sqrt{1-r^2 }\cot \left( 2\sqrt{1-r^2 } \eta \right)-r }  
\end{equation*}
\begin{equation*}
0 < \eta < \eta _{\infty } \ \ \ \ \ \ \ B(\eta _{\infty }) ^{-1} \equiv 0  \ \ \ \ \ r>0
\end{equation*}
Here we used the convention that
\begin{equation*}
\frac{ \sin \epsilon x}{\epsilon } = \left\{   \begin{array}{cc}
    \sin x & \epsilon = 1 \\ 
    x & \epsilon = 0 \\
    \sinh x & \epsilon = i \\ 
  \end{array} \right . 
  \end{equation*}
  and the corresponding expressions for $\cos x$.
 The original time coordinate $\tau $ is given by
 \begin{equation*}
 \frac{d\eta }{d\tau } = -r+\sqrt{1-r^2 }\cot \left( 2\sqrt{1-r^2 }\eta \right) .
 \end{equation*}
{\it Asymptotic behaviour } : 
\begin{eqnarray*}
\eta \rightarrow 0 & : & ds^2 = -dt^2 + t^{-\frac{2}{3}}dx^2 + t^{\frac{4}{3}}e^{-2cx}(dy^2 + dz^2 ) 
\end{eqnarray*}
\begin{eqnarray*}
\eta \rightarrow \eta _{\infty } & : & ds^2 = -dt^2 + t^2 \left[ dx^2 + e^{-2x}(dy^2 + dz^2 ) \right]  
\end{eqnarray*}

\end{description}

The vacuum and $\gamma =2$ solutions were discovered by Ftaclas and Cohen\cite{Ftaclas}; Roy and Singh\cite{Roy} gave the right-handed radiation solution. The other solutions appear to be new. It can be observed that all the left-handed solutions except $\gamma = \frac{2}{3}$ are asymptotically plane-waves in the future and the decay rate of electromagnetic energy density is given by $t^{-2}$. In terms of the average scale factor $l$, it can be expressed as $\rho _{em} \propto l^{-\frac{6\nu }{2+\nu }}$ using the parameter in (\ref{eq:nu}). When they are close to Milne, namely, in the isotropic limit $\nu \rightarrow 1$, the decay is as slow as $l^{-2}$. On the other hand, right-handed solutions which are attracted towards the Milne universe have a faster decaying rate of $l^{-6}$. The difference between the helicity is clearer when $\gamma =\frac{3}{2}$. Looking at the left- and right-orbit that have a same attractor in $\mathcal{F}(\tilde{h})$, their decaying rate always differs by $\frac{4}{r}$ with $r>1$. Generally speaking, as far as future asymptotic behaviour is concerned, left-handed solutions have a slower decaying rate than right-handed ones. 

\section{Comparison with the Perturbative Analysis}

Having shown the general stability of plane-waves and derived some detailed behaviour of the orbits in the LRS models, we now compare them to the results obtained in the perturbation around open FLRW. The type V and VII$_h$ models are anisotropic generalisations of open FLRW and their isotropic limit should exhibit some features seen in the perturbation of long wavelengths. We will see that the LRS modes appear to be generic and correspond to a super-adiabitc mode discussed by Barrow and Tsagas \cite{Tsagas}.

First, let us review the vector modes of electromagnetic perturbation around an open FLRW background. Following the convention of Barrow et.al. \cite{BMT}, the linearised Maxwell's equations in an arbitrary background frame are given by
\begin{eqnarray*}
\dot{B}_a &=& -2HB_a - {\rm curl }E_a , \\
\dot{E}_a &=& -2HE_a + {\rm curl }B_a .
\end{eqnarray*}
For the comparison with the analysis in the present article, we take the frame to be orthonormal and use the time coordinate $\tau $ introduced in the previous section, which is now identified with the conformal time of the Friedmann background. Note that the curl terms include the contribution from the isotropic background spatial curvature. Assuming $B_a(\tau, x^i ) = B_k (\tau ) Y^k _a(x_i) $ where a vectorial eigenfunction of the Laplace-Beltrami operator $Y^k _a $ satisfies
\begin{equation*}
(\Delta + \frac{k^2}{l^2} )Y^k _a = {\rm div}Y^k = \dot{Y}^k _a = 0,
\end{equation*}
the solutions are given by
\begin{equation}
B_k (\tau ) = \frac{\alpha }{l^2 }e^{\sqrt{2-k^2}\tau } +\frac{\beta }{l^2}e^{-\sqrt{2-k^2}\tau } \label{eq:mag}
\end{equation}
where $\alpha $ and $\beta $ are integration constants. The magnetic field follows the standard adiabatic decay law $B_a \propto l^{-2}$ for short-wavelength modes $k^2 \geq 2$, while the decay rate is reduced on large scales where the wave-numbers are smaller than the threshold $k^2 =2$ \cite{Tsagas}. For the vector perturbation, the magnetic field is always accompanied by electric field, which must be generated through the curl term. Substituting the magnetic solution (\ref{eq:mag}) into the Maxwell's equations and using the vector identity on hyperbolic space
\begin{equation*}
{\rm curl \ curl} = -\left( \Delta +\frac{2}{l^2}\right) + {\rm grad \ div } ,
\end{equation*}
we obtain
\begin{equation*}
E_a (\tau ,x_i ) = \frac{1}{\sqrt{2-k^2}}\left( \frac{\alpha }{l }e^{\sqrt{2-k^2}\tau }-\frac{\beta }{l }e^{-\sqrt{2-k^2}\tau }\right) {\rm curl }\ Y^k _a .
\end{equation*}
The apparent $l$-dependence of the electric field is misleading since the curl operator here includes another factor of $l^{-1}$. The physical amplitude of the electric field in the orthonormal frame behaves as $\propto \frac{e^{\pm \sqrt{2-k^2}\tau }}{l^2}$ as it should because of the symmetry between source-free electric and magnetic field. In terms of the energy density $\rho _{em} = \frac{1}{2}(B_a B^a + E_a E^a )$, the result is summarised as follows: There are two distinct behaviours above and below the threshold $k^2 =2$. For $k^2 \geq 2$, we just have the adiabatic decaying law $\rho _{em} \propto l^{-4}$. But for longer wavelengths, the decay rate is modified to $\rho _{em} \propto l^{-4} e^{\pm 2\sqrt{2-k^2 }\tau }$ each of which corresponds to growing ($\beta = 0$) and decaying ($\alpha = 0$) mode. If the background is Milne, then $l \propto e^{\tau }$ and we have 
\begin{equation}
\rho _{em} \propto l^{-4\pm 2\sqrt{2-k^2}}. \label{eq:decay}
\end{equation}

It later turns out to be helpful to compute the Poynting vectors associated with the solutions derived above. They are given by
\begin{equation*}
(E \times B)_a = \frac{-1}{\sqrt{2-k^2}}\left( \frac{\alpha ^2}{l^3} e^{2\sqrt{2-k^2}\eta } - \frac{\beta ^2}{l^3}e^{-2\sqrt{2-k^2}\eta } \right) (Y^k \times {\rm curl }\ Y^k )_a.
\end{equation*}
This means that for each given $Y^k_a$, the growing and decaying modes have exactly antiparallel Poynting vectors, which is reminiscent of the LRS models where the system shows different asymptotic behaviours according to the handedness of the field. 

Let us now consider the electromagnetic Bianchi models near isotropy. For the limit of isotropy, we assume $| \Sigma _{\pm , \times } | \ll 1$ in the dynamical system (\ref{eq:ein}) - (\ref{eq:pi}). This implies $N_{- ,\times}, \Pi _{\pm , \times }$ and $\Xi$ are all first order as well. The linearized electromagnetic sector obeys
\begin{eqnarray}
\Pi _{+}^{\prime } &=& 2(q-1)\Pi _{+} + 2A\Xi  \label{eq:P1}\\
\Xi ^{\prime } &=& 2(q-1)\Xi + 2A\Pi _{+} \label{eq:P2} \\
\Pi ^{\prime }_{-} &=& 2(q-1) \Pi _{-} \label{eq:P3} \\
\Pi ^{\prime }_{\times } &=& 2(q-1) \Pi _{\times } . \label{eq:P4}
\end{eqnarray}
First of all, we note that the evolution equations for $\Pi _{+}$ and $\Xi $ are closed among themselves and do not receive any contribution from $\Pi _{-}$ or $\Pi _{\times }$, which is exactly the feature of null electromagnetic field appearing in the LRS models. That is, any electromagnetic perturbation that must have nonzero energy density $\Pi _{+}$ inevitably excites the energy flux perturbation $\Xi $ as well, while $\Pi _{-}$ and $\Pi _{\times }$ are independent of $\Pi _{+}$ as the quadratic constraint (\ref{eq:pi}) is degenerate at linear order. The latter two modes are also subject to the remaining gauge freedom. Therefore their physical significance is questionable. Fortunately, we can consider the LRS modes $\Pi _{+}$-$\Xi $ separately as far as the isotropic limit is concerned. Moreover, the evolution of $\Pi _{-,\times }$ is trivial anyway since the equations (\ref{eq:P3}) and (\ref{eq:P4}) solve regardless of the background as
\begin{equation*}
\Pi _{-, \times } \propto H^{-2} l^{-4}
\end{equation*}
where $H$ is the background Hubble expansion rate. This is just the adiabatic decaying law of radiation. From now on, therefore, we focus our attention on the LRS models for which we have some exact solutions and full dynamical picture.

For $\frac{2}{3} <\gamma \leq 2$, the generic behaviour of the LRS models is to start from an anisotropic singularity dominated by extrinsic curvature, possibly becoming almost isotropic for an intermediate time interval and attracted towards electro-vacuum, which is driven by the intrinsic curvature synchronised with the null Maxwell field in the future. The handedness of the electromagnetic wavelet, namely the direction of the Poynting vector, is important in deciding the asymptotic states. In left-handed cases, both past and future dynamics receive significant contributions from the Maxwell field appearing as plane-wave space-times while right-handed fields only affect the intermediate evolution. 

In the cosmological context, we are interested in the future asymptotic behaviour of the electromagnetic energy density $\rho _{em}$ in the isotropic limit. That is, when an orbit is attracted towards Milne point $M(\tilde{h})$. This is also the situation where the magnitude of the shear variable $|\Sigma _{+}|$ takes its minimum in most cases: i.e. all the right-handed models and left-handed models with $\gamma \leq  \frac{4}{3}$ (FIG. \ref{fig:rad}). 
\begin{figure}
\includegraphics[scale=0.20]{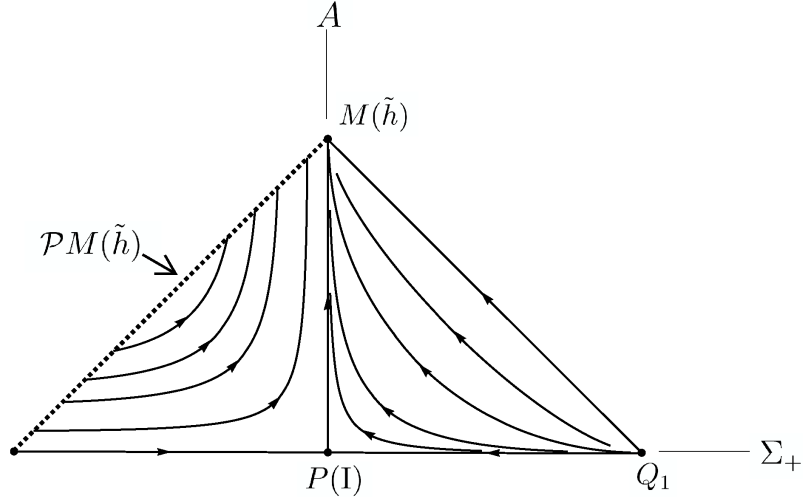}
\caption{\label{fig:rad}Phase portrait of type V and VII$_h$ LRS models for radiation ($\gamma = \frac{4}{3}$). All the orbits achieve their most isotropic state at late time.}
\end{figure}
We can solve the perturbed LRS Maxwell's equations (\ref{eq:P1}) and (\ref{eq:P2}) with $q=0, A=1$ and obtain the two independent solutions
\begin{equation*}
\Pi _{+} \propto \frac{\rho _{em}}{H^2 } \propto l^{-4} + {\rm const.}.
\end{equation*}
They translate into the decaying mode $\rho _{em} \propto l^{-6}$ and the growing mode $\rho _{em} \propto l^{-2}$. Looking at the exact solutions for $\gamma = \frac{4}{3}$ and $\gamma =2$, we can easily see that the decaying (growing) mode represents the right-handed (left-handed) orbits. In fact, this is the generic feature of asymptotic behaviour of the LRS models in isotropic limit. 
Near isotropy, the decaying rate of electromagnetic energy density is decided by the handedness of the field. Recalling the perturbative result for open FLRW (\ref{eq:decay}), the long-wavelength inhomogeneous perturbation of $k^2 =1$ shares the decay rates and the relation to the handedness with the LRS modes studied here. This appears to imply that the super-adiabitc mode of $k^2 =1$ in open FLRW generalizes to the left-handed LRS orbit attracted towards the close-to-isotropy plane-wave solutions.

We can see another isotropic asymptotic behaviour for $\gamma = \frac{2}{3}$ (FIG. \ref{fig:soft}).  The future-asymptotic state of the Bianchi model is $\mathcal{F}(\tilde{h})$ and therefore we should compare it with the perturbation around the one-parameter family of open FLRW backgrounds $\mathcal{F}(\tilde{h})$ instead of the Milne universe. 
\begin{figure}
\includegraphics[scale=0.20]{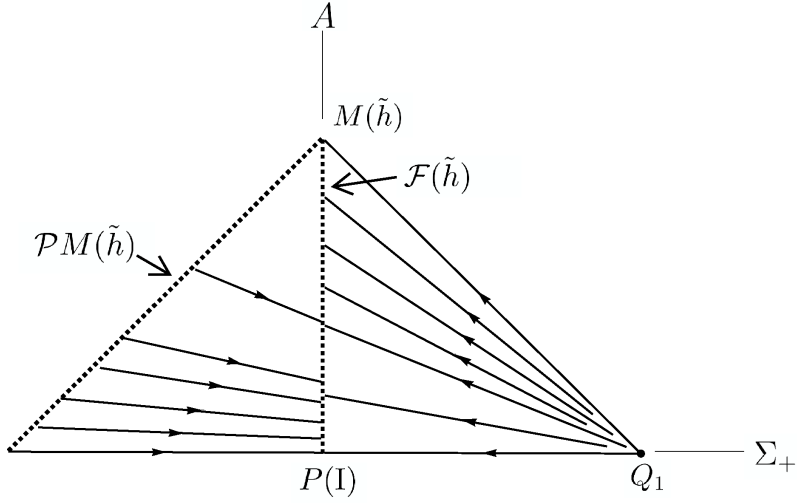}
\caption{\label{fig:soft}Phase portrait of type V and VII$_h$ LRS models for $\gamma = \frac{2}{3}$. The orbits are attracted towards the line of equilibrium points $\mathcal{F}(\tilde{h})$ indicated by the vertical dotted lines.}
\end{figure}
Looking at the LRS exact solutions, we see that the electromagnetic energy density goes as $l^{-4-\frac{2}{r}}$ for right-handed solutions and $l^{-4+\frac{2}{r}}$ for left-handed ones with $r>1$. Here $r$ parametrizes different members within the family of $\mathcal{F}(\tilde{h})$ as well as different obits. On the other hand, the inhomogeneous perturbation around the FLRW backgrounds $\mathcal{F}(\tilde{h})$, labeled by the same parameter $r$, gives the growing mode
\begin{equation*}
\rho _{em} \propto \frac{e^{2\sqrt{2-k^2}\tau }}{l^4} \propto l^{-4 + \frac{2\sqrt{2-k^2}}{r}}
\end{equation*}
and the corresponding decaying mode is
\begin{equation*}
\rho _{em} \propto \frac{e^{-2\sqrt{2-k^2}\tau }}{l^4} \propto l^{-4-\frac{2\sqrt{2-k^2}}{r}}.
\end{equation*}
Again, we can identify the $k^2=1$ growing mode with left-handed asymptotic behaviour and the decaying mode with the right-handed one. The presence of the isotropic perfect fluid smears out the effect of magneto-curvature coupling. 

It should be mentioned that these arguments are far from conclusive since the orbits that are future-asymptotic to Milne are of measure zero in the general electromagnetic class B models. On the other hand, there is no LRS left-handed orbit which approaches Milne and we can only have approximately isotropic final states there while all the right-handed orbits approach $M(\tilde{h})$. In this sense, the above comparison is asymmetric between left and right, in turn growing and decaying modes. However, it simply indicates that the nonlinear interaction can lead to a quite different behaviour from linear perturbations and the dominant effect their seems to be the electromagnetic plane-wave mode. It is reasonable that this nonlinear plane-wave corresponds to the growing mode of the perturbation while the decaying mode can only be seen in the exceptional orbits which are asymptotic to Milne.

\section{Concluding Remarks}
We have investigated a class of spatially homogeneous Einstein-Maxwell space-times and described possible asymptotic behaviours by a dynamical systems analysis. The results are analogous to pure-gravitataional models, with the electromagnetic field acting as a kind of bridge between extrinsic and intrinsic-curvature dominated regimes. The extended electromagnetic plane-waves are stable attractors of the system. In the LRS models, we derived more detailed features of the dynamics by looking at some exact solutions. The handedness of the null field plays a crucial role in the magneto-curvature coupling in the restricted class of models. These LRS null Maxwell modes appear to generalise the electromagnetic vector perturbations around open FLRW with the wave-vector $k=1$. 

The Bianchi models provide another example of the close relation between gravity and electromagnetism. The dynamics here are surprisingly simple considering the dimensionality of the system. It is interesting to note that Maxwell fields can dominate over perfect fluids, for example dust at late times in a long-wavelength limit, even though the simple adiabatic decaying law of electromagnetic energy density suggests otherwise.

As to the connection to the perturbations, this analysis shows that we might not necessarily be able to ignore the vector mode with the wave-number smaller than $\sqrt{2}$ in open FLRW models as used to be done (e.g. Ref. \onlinecite{Goode}). It was already argued for scalar modes that we should take into account the supercurvature mode $k<1$ when we are concerned with a random distribution of perturbations over the sky, even though any causal perturbation could be described by the subcurvature modes which span the basis of square-integrable functions over hyperbolic space \cite{Lyth}. In the case of the vector mode, it is already not obvious what square-integrable means and it is not clear which wavelengths we should include in what situation. The results here are of interest regarding this issue because homogeneous models seem to correspond to neither 
 $k = \sqrt{2}$ nor $k=0$, but $k=1$.

\begin{acknowledgements}
The author would like to thank Professor John D. Barrow for initiating this work and for his help on the course of the analysis. The author would also like to thank Dr. Stephen Siklos and Dr. Anthony Challinor for their encouragement and Professor Christos Tsagas for useful discussions. The author is supported by the Cambridge Overseas Trust.
\end{acknowledgements}

\appendix*

\section{Class B metrics in automorphism variables}
For reference, we give a set of metric variables, proposed by Siklos \cite{Siklos} and their relation to the orthonomal frame. First of all, we introduce canonical basis vectors $\{ \bold{E}_{\alpha } \} $ defined by
\begin{eqnarray*}
\bold{E}_1 = \partial _x  \ \ \ \ \ && \bold{E}_2 = e^x \left( \cos \epsilon kx \partial _y + \epsilon \sin \epsilon kx \partial _z \right) \\ 
&& \bold{E}_3 = e^x \left( -\epsilon ^{-1}\sin \epsilon kx \partial _y + \cos \epsilon kx \partial _z \right) 
\end{eqnarray*}
Here the parameter $\epsilon $ is taken to be $1$ for VII$_h$, $0$ for V and $i$ for VI$_h$. In this coordinate system, the metric tensor can be written in the following form:
\begin{equation*}
g^{\alpha \beta } = \tilde{g}^{\alpha ^{\prime }\beta ^{\prime }}(A\Phi B)_{\alpha ^{\prime }}^{\ \alpha }(A\Phi B)_{\beta ^{\prime }}^{\ \beta },
\end{equation*}
where
\begin{equation*}
\tilde{g} =  \left(  \begin{array}{ccc}
    a^2 & 0 & 0 \\ 
    0 & e^{2\mu } & 0 \\ 
    0 & 0 & e^{-2\mu } \\ 
  \end{array} \right) , \ \ \ \ \ \Phi  =  \left( \begin{array}{ccc}
    1 & 0 & 0 \\ 
    0 & \cos \epsilon \phi & \epsilon \sin \epsilon \phi \\ 
    0 & \epsilon ^{-1} \sin \epsilon \phi & \cos \epsilon \phi \\ 
  \end{array} \right) ,
\end{equation*}
\begin{equation*}
A =   \left( \begin{array}{ccc}
    1 & 0 & 0 \\ 
    0 & e^{-\lambda } & 0 \\ 
    0 & 0 & e^{-\lambda } \\ 
  \end{array} \right) , \ \ \ \ B =  \left( \begin{array}{ccc}
    1 & b_2 & b_3 \\ 
    0 & 1 & 0 \\ 
    0 & 0 & 1 \\ 
  \end{array} \right).
  \end{equation*}
$\Phi, A$ and $B$ are called automorhpism transformations since if $\{ \bold{E}_{\alpha } \} _{\alpha = 1,2,3}$ satisfy the class B structure equations , $\bold{X}_{\alpha ^{\prime }} \equiv (A \Phi  B)_{\alpha ^{\prime }}^{\ \alpha } \bold{E}_{\alpha } $ also satisfy the same commutation relations.

By construction, 
\begin{equation*}
\bold{e}_1 \equiv a \bold{X}_1 , \ \ \  \bold{e}_2 \equiv e^{\mu } \bold{X}_2 , \ \ \ \bold{e}_3 \equiv e^{-\mu }\bold{X}_3 
\end{equation*}
are orthonormal. Introducing reduced variables
\begin{eqnarray*}
B_2 \equiv  -\frac{1}{2}e^{\lambda -\mu } \left( \dot{b}_2 \cos \epsilon \phi + \dot{b}_3 \epsilon \sin \epsilon \phi \right) , \ \ \ B_3 \equiv \frac{1}{2}e^{\lambda +\mu } \left( \dot{b}_2 \epsilon ^{-1} \sin \epsilon \phi - \dot{b}_3 \cos \epsilon \phi \right),
\end{eqnarray*}
the Ricci tensor in the orthonormal frame is given as follows:
\begin{eqnarray*}
 R_{00} &=& \frac{\ddot{a}}{a} -2\frac{\dot{a}^2}{a^2} -2a^2 (B_2^2 +B_3^2 ) -2\ddot{\lambda }-2\dot{\lambda }^2 -2\dot{\mu }^2 -\dot{\phi }^2(\cosh 4\mu -\epsilon ^2 ), \\
 R_{01} &=& 2\dot{a}+2a \dot{\lambda }-\epsilon ^2 ka\dot{\phi }(\cosh 4\mu -\epsilon ^2 ) , \\
 R_{02} &=& -3a^2 B_2 - \epsilon ^2 ka^2 B_3 e^{2\mu } ,  \ \ \ R_{03} \  =  \ -3a^2 B_3 + ka^2 B_2 e^{-2\mu } , \\
 R_{11} &=& -\frac{\ddot{a}}{a}+2 \frac{\dot{a}^2}{a^2} +2a^2(B_2^2 +B_3^2)-2a^2 -2\frac{\dot{a}}{a}\dot{\lambda }-k^2a^2(\cosh 4\mu -\epsilon ^2 ) , \\
 R_{12} &=& \dot{(aB_2)}+aB_2 (3\dot{\lambda }-\dot{\mu })-\epsilon ^2 aB_3 \dot{\phi }e^{2\mu } \ \ \ R_{13} \ = \  \dot{(aB_3)} +aB_3 (3\dot{\lambda }+\dot{\mu }) +aB_2 \dot{\phi }e^{-2\mu } , \\
 R_{22} &=& \ddot{\lambda }-\ddot{\mu }+(\dot{\lambda }-\dot{\mu })\left( 2\dot{\lambda }-\frac{\dot{a}}{a} \right) -2a^2 -2a^2 B_2 ^2+2(\dot{\phi }^2 -k^2 a^2 )\cosh 2\mu \sinh 2\mu , \\
 R_{33} &=& \ddot{\lambda }+\ddot{\mu }+(\dot{\lambda }+\dot{\mu }) \left(2\dot{\lambda } -\frac{\dot{a}}{a} \right) -2a^2 -2a^2 B_3^2 -2(\dot{\phi }^2 -k^2 a^2 )\cosh 2\mu \sinh 2\mu , \\
 R_{23} &=& \left( -\ddot{\phi }+\frac{\dot{a}}{a}\dot{\phi } \right) \frac{e^{2\mu }-\epsilon ^2 e^{-2\mu }}{2} -2a^2 B_2 B_3 - \epsilon ^2 (ka^2 -\dot{\lambda }\dot{\phi })(e^{2\mu }-\epsilon ^2 e^{-2\mu }) -2\dot{\mu }\dot{\phi }(e^{2\mu }+\epsilon ^2 e^{-2\mu }) .
\end{eqnarray*}
The LRS specialisation in section 5 is obtained by setting $\phi = \mu =0$.

\end{document}